\documentclass[final,a4paper,UKenglish,cleveref, autoref, thm-restate]{lipics-v2021}
\pdfoutput=1 %
\nolinenumbers

\usepackage{framed}
\usepackage{tikz}
\usepackage[normalem]{ulem}
\usetikzlibrary{automata,calc}

\hideLIPIcs  %

\newcommand{\NN}{\mathbb{N}}

\title{Natural Colors of Infinite Words}

\author{Rüdiger Ehlers}{Clausthal University of Technology, Germany \and \url{https://www.ruediger-ehlers.de/} }{ruediger.ehlers@tu-clausthal.de}{https://orcid.org/0000-0002-8315-1431}{%
This work was supported by the German Science Foundation (DFG) under Grant No.~322591867.
}%

\author{Sven Schewe}{University of Liverpool, United Kingdom\and \url{https://cgi.csc.liv.ac.uk/~sven/}}{sven.schewe@liverpool.ac.uk}{https://orcid.org/0000-0002-9093-9518}{%
\includegraphics[height=8pt]{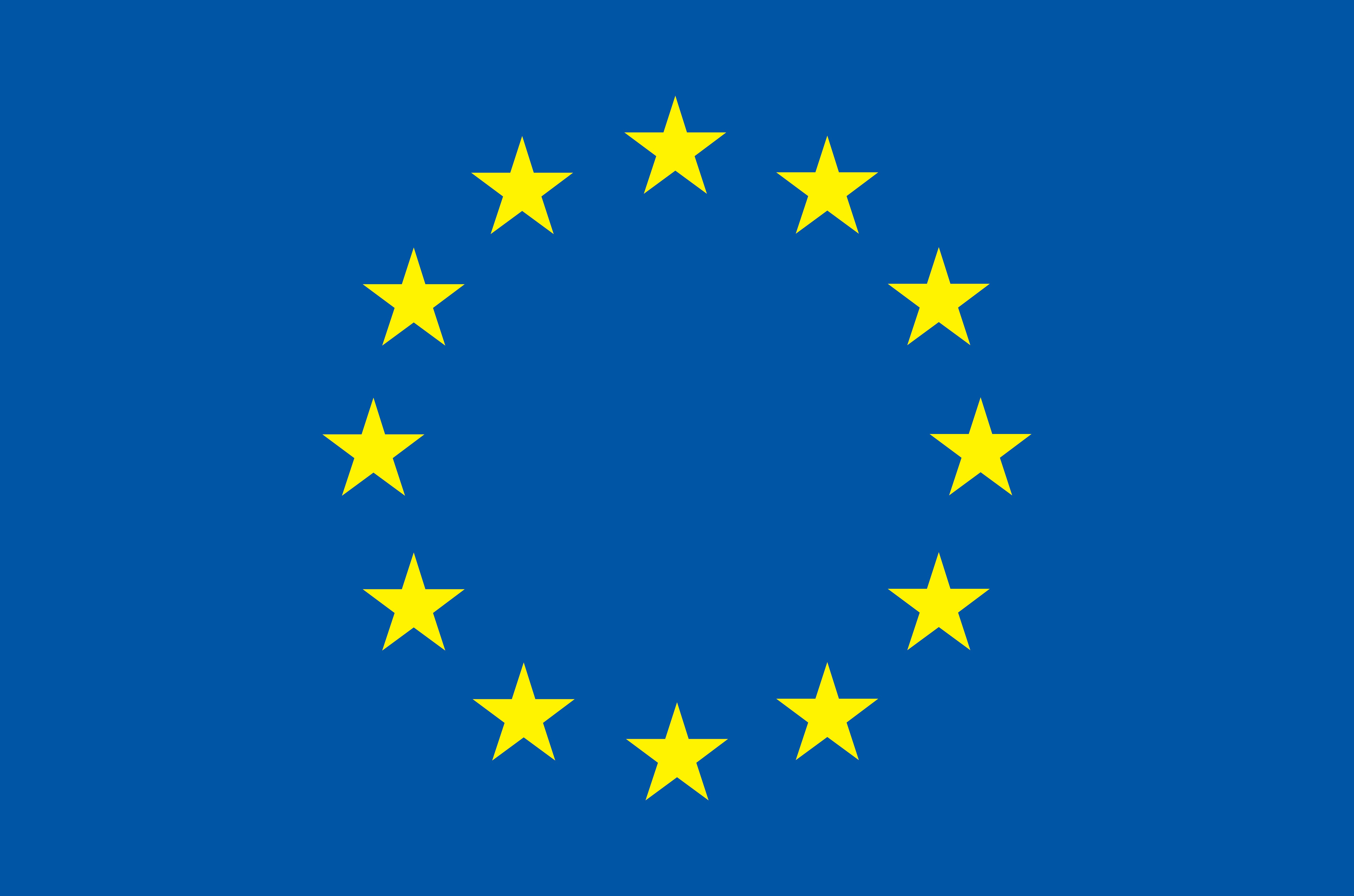} 
This project has received funding from the European Union’s Horizon 2020 research and innovation programme under grant agreement 956123 (FOCETA).%
}

\authorrunning{R\"udiger Ehlers and Sven Schewe} %

\Copyright{Rüdiger Ehlers and Sven Schewe} %



\category{} %

\relatedversion{} %

\acknowledgements{The first author likes to thank Arved Friedemann for interesting discussions that had a positive influence on the undertaking of this work.}%

\begin{document}

\maketitle

\begin{abstract}
While finite automata have minimal DFAs as a simple and natural normal form, deterministic omega-automata do not currently have anything similar. One reason for this is that a normal form for omega-regular languages has to speak about more than acceptance -- for example, to have a normal form for a parity language, it should relate every infinite word to some natural color {{for this} language}. 
This raises the question of whether or not a concept such as a natural color of an infinite word {({for a given} language)} exists, and, if it does, how it relates back to automata.

We define the natural color of a word purely based on an omega-regular language, and show how this natural color can be traced back from any deterministic parity automaton after two cheap and simple automaton transformations. The resulting \emph{streamlined} automaton does not necessarily accept every word with its natural color, but it has a `co-run', which is like a run, but can \emph{once} move to a language equivalent state, whose 
color is the natural color, and no co-run with a higher color exists.

The streamlined automaton defines, for every color $c$, a good-for-games co-Büchi automaton that recognizes the words whose natural colors {w.r.t.~the represented language} are at least $c$. This provides a canonical representation for every $\omega$-regular language, because good-for-games co-Büchi automata have a canonical minimal---and cheap to obtain---representation for every co-B\"uchi language.
\end{abstract}

\section{Introduction}
\label{sec:introduction}

A classical question in the theory of automata is how to define \emph{canonical representations} of regular languages.
Such a representation, typically in the form of an automaton, for a language has several advantages.
For once, different canonical automata must define different languages.
But reasonably defined canonical automata are also concise, and normally a minimal (and thereby natural) representative of all language equivalent automata of the same type, which makes them a natural representative of the language they recognize.

Such definitions of canonicity build on---or deliver---insights into the possible structure of an automaton for a given language. For instance, canonical deterministic automata over finite words have exactly one state per (reachable) suffix language, and
the Myhill-Nerode automaton minimization procedure is able to translate every deterministic automaton over finite words into its canonical form in polynomial time \cite{HOPCROFT1971189}.
The concepts that underpin insightful canonicity definitions often give rise to efficient minimization procedures, which makes it attractive to apply them in practical applications.
In turn, such concepts are useful in concisely defining a language, including for practical applications like learning, model checking, or synthesis.

For regular languages over infinite words, obtaining insightful canonical forms has remained a challenge.
Such languages are useful for reasoning about \emph{reactive systems}, i.e., computational systems that continuously read from an input stream while producing an output stream.
While Löding~\cite{DBLP:journals/ipl/Loding01} gave a construction for computing canonical and minimal deterministic \emph{weak} automata, which can encode some such languages, this automata class is very restricted in that in every strongly connected component, either all states are accepting or all states are rejecting.
This means that simple languages such as `there are (in)finitely many $a$s in the word' cannot be represented by them.

After the result by Löding \cite{DBLP:journals/ipl/Loding01}, there was, for quite a while, little progress on canonical forms for more expressive subclasses of $\omega$-regular languages. This is partially rooted in the fact that deterministic Büchi (and co-Büchi) automata---which are among the simplest $\omega$-automaton types and {cannot even} capture all $\omega$-regular languages---have an NP-complete minimization problem \cite{DBLP:conf/fsttcs/Schewe10}.
This shows that, unlike in the case of languages over finite words, deterministic automata cannot be used for defining a canonical form that is easy to compute.

Only very recently, Abu Radi and Kupferman \cite{DBLP:conf/icalp/RadiK19} observed that, via a slight generalization from deterministic co-Büchi automata to good-for-games \cite{DBLP:conf/csl/HenzingerP06} co-Büchi automata, we obtain an automaton model for co-Büchi languages that permits a polynomial-time minimization procedure;
this gives rise to an insightful canonical form.
Good-for-games co-Büchi automata (and similarly, good-for-games parity automata with a fixed set of colors) are not more expressive than deterministic automata with the same acceptance condition \cite{DBLP:conf/csl/HenzingerP06}, but they can be more concise.
Interestingly, this added conciseness is what enables polynomial-time minimization and thereby efficiently computing canonical automata.
{In the canonical minimal automata computed using the construction by Abu Radi and Kupferman \cite{DBLP:conf/icalp/RadiK19}, non-determinism only appears along \emph{rejecting} transitions that connect different strongly connected components {that consist} only of accepting transitions. Hence, the different deterministic strongly connected components represent the different ways in which a word can be accepted and hence provide insight into the structure of the represented language.}

{The result by Abu Radi and Kupferman} raises the question whether this result can be extended to obtain a canonical {and insightful} representation for general $\omega$-regular languages.
Such an extension would intuitively need to use a richer type of acceptance condition than co-Büchi acceptance, as co-Büchi acceptance is too limited in expressivity.
The weakest acceptance condition that offers full $\omega$-regularity in this context is \emph{parity acceptance}.
In parity automata, a word is accepted if, and only if, the lowest \emph{color} that occurs infinitely often along a run of the automaton is even. 

We could salvage the polynomial-time canonicization procedure for co-Büchi acceptance while using a parity-type acceptance condition by representing an $\omega$-regular language $L$ by a falling chain of languages $L_0 \supset L_1 \supset L_2 \ldots \supset L_c$ s.t.\
$L_0$ is 
the universal language 
and each language $L_i$ is a co-Büchi language.
{A word is then accepted by the \emph{chain} of languages if and only if the highest $i$ such that the word is in $L_i$ is even.}
As all of these languages are co-B\"uchi languages, 
we can represent each of them by their canonical minimal good-for-games automaton, such that, together, these automata are a canonical representation of $L$.

{The crucial piece that is currently missing in the literature to obtain such a canonical representation of an arbitrary $\omega$-regular language, however, is which word should be 
in which language $L_i$.
Omega-regular languages can be decomposed into such chains in different ways, and for the overall chain to be a canonical representation of the language, we need a fixed way for decomposing the language $L$ into co-Büchi languages $L_i$.
In other words, we are missing a definition of the \emph{natural} color of a word that defines the highest index $i$ such that the word is in $L_i$. 
{This natural color depends} on the overall language to be represented, as this color reflects where in the decomposition of a given language the word resides.
}

For a useful canonicity definition, we need the allocation of words to the individual $L_i$ for a given $\omega$-regular language $L$ to have several properties:
\begin{enumerate}
\item It should be based on the language $L$ alone, and be independent of the syntactic structure of any representation of it (such as some parity automaton that recognizes $L$),
\item it should be easy to compute for a given word and a given representation of $L$ (such as a deterministic parity automaton), and
\item starting from an automaton representation of $L$, the sizes of co-Büchi automata for the {languages $L_i$} should be \emph{small}, ideally not bigger than the size of an automaton for $L$.
\end{enumerate}

In this paper, we provide a definition of a \emph{natural} color of an infinite word for a given $\omega$-regular language that has these properties.
Our definition distills the idea that, in a parity automaton, only the lowest color visited infinitely often along a run matters, into a concept that can be defined on languages alone, without referring to a specific automaton. {We then use this for introducing a canonical representation of arbitrary $\omega$-regular languages as a chain of co-Büchi languages. While the definition of the natural color of a word (for a given language) is the main technical contribution of this paper, its study is motivated by what it can be used for, namely for establishing a canconical representation for $\omega$-regular languages, which is the second contribution of this paper.}

{We show that our particular definition of the natural color of a word (for a given language)} has the property that every deterministic parity automaton can be translated into a form from which the natural color of a word can easily be read off.
This works in two steps: 
We first simplify an automaton by ordering its strongly connected components (using an order that respects reachability) and bending all transitions to language equivalent states in the maximal component they reside in (besides removing unreachable and unproductive states). In a second step, we construct a so-called \emph{streamlined} form of a parity automaton that retains the transition structure.
Both transformations are tractable.

From a streamlined automaton, we can furthermore obtain, again in polynomial time, good-for-games co-Büchi automata for all languages $L_i$. %
They are no larger than the original streamlined parity automaton, and therefore no larger than the deterministic parity automaton we started with.
Moreover, they can subsequently be minimized \cite{DBLP:conf/icalp/RadiK19} to obtain a canonical representation of a given $\omega$-regular language.
{This minimization can also yield an exponential advantage over a representation as deterministic co-B\"uchi automata \cite{DBLP:conf/icalp/KuperbergS15}.}

As a consequence, with our definition, one can obtain, in polynomial time, a canonical representation of the language of a deterministic parity automaton.
While this representation is not a single automaton, deviating from deterministic branching was necessary in order to avoid the NP-hardness of minimizing deterministic parity automata. 
Furthermore, it was shown that good-for-games parity automata are also NP-hard to minimize \cite{DBLP:conf/fsttcs/Schewe20}, so a further generalization had to be made. 
By choosing a sequence of good-for-games co-Büchi automata as this generalization, we avoid introducing a more complex automaton type at the cost of having multiple automata.

{While it is possible  to define other variants of what the natural color of a word (for a given language) could be, our definition has the advantage that it coincides with the color of a word of \emph{some} parity automaton for a given language while permitting a translation from a deterministic parity word automaton to a canonical representation of its language in polynomial time. }

\section{Preliminaries}

Given a set $S$, we denote the set of finite sequences (words) of elements in $S$ as $S^*$ and the set of infinite sequences of elements in $S$ as $S^\omega$. Sets of words are also called \emph{languages} (over some alphabet). We only consider finite alphabets. 
We denote the set of natural numbers including $0$ by $\mathbb{N}$. Given a language $L \subseteq \Sigma^\omega$ and a finite word $w \in \Sigma^*$, we define the suffix language of $L$ over $w$ as $\mathcal{L}^\mathit{suffix}(L,w) = \{w' \in \Sigma^\omega \mid ww' \in L\}$.

We define \emph{parity automata with transition-based acceptance} as tuples $\mathcal{A} = (Q,\Sigma,\delta,Q_0)$, where 
$Q$ is a finite set of states,
$\Sigma$ is a finite alphabet,
$\delta \subseteq Q \times \Sigma \times Q \times \NN$ is a transition relation, and $Q_0 $ is the set of initial states.
We say that a word $w = w_0 w_1 \ldots \in \Sigma^\omega$ induces a run $\pi = q_0 q_1 \ldots$ of the automaton with the corresponding color sequence $\rho = \rho_0 \rho_1 \ldots \in \mathbb{N}^*$ if $q_0 \in Q_0$ and, for every $j \in \NN$, we have $(q_j,w_j,q_{j+1},\rho_j) \in \delta$. 
We say that $w$ is \emph{accepted} by $\mathcal{A}$ if there exists a run $\rho$ for the word on which the lowest color that occurs infinitely often along $\rho$ is even.

We say that $\mathcal{A}$ is deterministic if, for every state $q \in Q$ and $x \in \Sigma$, we have exactly one element of the form $(q,x,q',c) \in \delta$, and $Q_0$ is a singleton set.
We {henceforth} use $q_0$ as tuple element for the initial state for deterministic automata.
In deterministic automata, every word induces a unique run/color sequence combination. This also allows us to define, by slight abuse of notation, for each state $q$ and word $w \in \Sigma^*$, $\delta(q,w)$ to be the unique state reached from $q$ after reading $w$.
We refer to the smallest color that occurs infinitely often in the color sequence corresponding to a run for $w$ as the \emph{color of $w$ in $\mathcal{A}$}.
We also call it the \emph{dominating color} among the ones occurring infinitely often in the color sequence for $w$. 

The set of words accepted by $\mathcal{A}$ is called its language, also denoted by $\mathcal{L}(\mathcal{A})$.
We say that $\mathcal{A}$ is a co-Büchi automaton if only the colors $1$ and $2$ occur along transitions. Transitions with color $1$ and $2$ are also called \emph{rejecting} and  \emph{accepting} transitions, respectively.
The automaton $\mathcal{A}_q$ with $q \in Q$ denotes a variant of $\mathcal{A}$ for which the initial state has been replaced by $q$.
{
We say that two states $q,q' \in Q$ of a deterministic parity automaton (\emph{DPA}) $\mathcal A$ are \emph{equivalent}, denoted by $q \sim_\mathcal{A} q'$, if, and only if, they have the same language $\mathcal L(\mathcal A_q) = \mathcal L(\mathcal A_{q'})$.}

We say that an automaton $\mathcal{A}$ is \emph{good-for-games} if there exists a strategy function \cite{DBLP:conf/icalp/RadiK19} $f : \Sigma^* \rightarrow Q \times \NN$ such that for each word $w = w_0 w_1 \ldots \in \Sigma^\omega$ in the language of $\mathcal{A}$, there exists an accepting run $\pi = q_0 q_1 \ldots \in Q^\omega$ with corresponding color sequence $\rho = \rho_0 \rho_1 \ldots \in \NN^\omega$ for it such that for all $j \in \NN$, we have {$(q_{j+1},\rho_j) = f(w_0 \ldots w_j)$}. Note that such a run is unique for each $w$.

Given an automaton $\mathcal{A} = (Q,\Sigma,\delta,Q_0)$,  we say that some state set $\tilde Q \subseteq Q$ is a \emph{strongly connected component} (SCC) if, for each $q,q' \in \tilde Q$, we have that there exists a sequence of states $q_1, \ldots, q_n$ all in $\tilde Q$ for some $n \in \NN$ such that $q_1 = q$, $q_n = q'$, and for every $1 \leq j < n$, there exist $x \in \Sigma$ and $c \in \NN$ such that $(q_j,x,q_{j+1},c) \in \delta$. We  say that $\tilde Q$ is a \emph{maximal SCC} if $\tilde Q$ cannot be strictly extended by any set of states such that the resulting state set is still an SCC.
Transitions that can only be taken once in a run (starting from any state) are called \emph{transient}, and these connect different SCCs. {We also say that a state is \emph{transient} if it can only be visited once along a run.

Given a co-Büchi automaton $\mathcal{A} = (Q,\Sigma,\delta,Q_0)$, we say that some state set $\tilde Q \subseteq Q$ {with some set of transitions $\tilde \delta \subseteq \delta \cap \tilde Q \times \Sigma \times \tilde Q \times \{2\}$} is an \emph{accepting SCC} if {$(\tilde Q,\tilde \delta)$} is an SCC.
Similarly,  {$(\tilde Q,\tilde \delta)$} is called a \emph{maximal accepting SCC} if {$(\tilde Q,\tilde \delta)$ cannot be extended with further states / accepting transitions} without losing the property that it is an accepting SCC.
}

\section{Towards A Canonical Language Representation}
\label{sec:TowardsCanonicalRepresnetaion}

The core definition we provide in this paper, namely the \emph{natural color of a word}, lifts the idea of parity acceptance from automata to languages.
{Such a natural color will always be defined with respect to a given language, but for the brevity of presentation, we will not always mention this language henceforth.}

Since, at the level of languages, there are no colors of transitions that can be referred to, the definition of the natural color of a word needs to capture the idea of colors in a way that does not employ the colors of transitions. 

In this section, we make some observations on why some languages need a certain number of colors in a parity automaton, and identify ways, in which we can abstract from the automaton representation along the way.
We then distil the observations to define the natural color of a word in the next section.

Niwińkski and Walukiewicz \cite{DBLP:conf/stacs/NiwinskiW98} have given a polynomial-time algorithm to minimize the number of different colors in a parity automaton.
While their algorithm targets parity automata in which states---rather than transitions---are labeled with colors, it is not difficult to extend it to transition-based acceptance. 

The core idea used in their algorithm is that, in order for a parity automaton that recognizes a language to need at least $n$ colors, there needs to exist a so-called \emph{flower} with at least $n$ colors.
Such a flower is defined to satisfy two properties.
It firstly has a sequence of colors $c_1 < c_2 < \ldots < c_n$ such that every two successive colors in the sequence alternate by whether they are even or odd.
Secondly, there exists a center state such that, for each color $c_i$, there exist paths from the center state back to itself such that the dominating color occurring along the transitions along the path is $c_i$.
Following the terminology of Niwińkski and Walukiewicz, we refer to such paths as \emph{flower loops}.
Figure~\ref{fig:flowerVersionB} shows an example parity automaton that contains such a flower over the colors 1, 2, 3, 4, and 5. 

Niwińkski and Walukiewicz have shown that no parity automaton with fewer than $n$ colors can encode a language that admits a flower with $n$ colors  \cite{DBLP:conf/stacs/NiwinskiW98}.
This is because a flower defines a hierarchy over words that are, alternatingly, accepted or rejected by an automaton, and the $n$ different colors are needed to detect on which level in the hierarchy a word is located.

Figure~\ref{fig:wordsInTheHierachy} shows such a hierarchy of words for the parity automaton from Figure~\ref{fig:flowerVersionB}. 
They all have in common that the state $q_c$ in the center of the flower is visited infinitely often in a run of the automaton for the respective word. 

The first word, $(ca)^\omega$, leads to only transitions with color $5$ being taken, so the color of the word is $5$, and the word is rejected. 

The second word is built by injecting $bb$ strings at positions in the word, at which the respective run is in the center state $q_c$. Note that $bb$ causes transitions $q_c \xrightarrow{b} q_2 \xrightarrow{b} q_c$ in the run for the second word, so that the added string causes an excursion in the run that leads back to the same state.
In this way, the run of the second word can be obtained from the run for the first word by adding elements at the positions in which a finite substring is inserted into the word.
Because in the run for the second word, color $4$ is visited infinitely often, this becomes the color of the modified word.

The other words are built according to the same idea: 
By injecting finite words leading from $q_c$ back to $q_c$ (while taking a transition with a lower color) into the word at positions in the word on which a run for the word is at state $q_c$ anyway, we obtain a new word that is accepted by the automaton if, and only if, the old word is rejected.

The most significant color (here, 1) now has the special property that inserting more {loops} %
does not change whether or not the word is accepted, as the most significant color {of the automaton} is already visited along the run.
We can formalize this as follows:
Let $w = w_0 w_1 \ldots$ be a word and $\pi = q_0 q_1 \ldots$ be a run of the automaton over $w$.
We say that a finite word $w'$ is a state-invariant injection at a position $i \in \mathbb{N}$ in the word if $\delta(q_i,w') = q_i$\label{used:extensionofDeltaInParityAutomaton}.
We can characterize the words $w$ that are recognized with the most significant color in a strongly connected component (SCC) of the parity automaton as those for which every word $\tilde w$ that results from an infinite sequence of injections of state-invariant words into $w$ is accepted by the automaton if, and only if, $w$ is accepted by the automaton. 
Note that the restriction to strongly connected components is necessary as other parts of a parity automaton may employ more colors, and hence this property holding for a word whose run ends in some SCC does not exclude that some other automaton part uses more colors, including a lower one.

\begin{figure}
    \centering
    \begin{tikzpicture}
    \node[shape=circle,draw,thick] (qc) at (0,0) {$q_c$};
    \node[shape=circle,draw,thick] (q3) at (5,0) {$q_3$};
    \node[shape=circle,draw,thick] (q2) at (3,2) {$q_2$};
    \node[shape=circle,draw,thick] (q1) at (-3,2) {$q_1$};
    
    \draw[thick,->] (qc) to[bend left=10] node[below left] {a/3} (q1);
    \draw[thick,->] (qc) to[bend left=10] node[above left] {b/4} (q2);
    \draw[thick,->] (qc) to[bend left=10] node[above right,pos=0.7] {c/5} (q3);
    \draw[thick,->] (q1) to[bend left=10] node[below] {b,c/2} (q2);

    \draw[thick,->] (q3) to[bend left=10] node[below left] {a,b,c/5} (qc);

    \draw[thick,->] (q1) to[bend left=10] node[above right] {a/1} (qc);
    
    \draw[thick,->] (q2) to[bend left=10] node[above right] {a,c/3} (q3);
    \draw[thick,->] (q2) to[bend left=10] node[below right,pos=0.2] {b/5} (qc);
    
    \draw[thick,->] ($(qc)+(-1,-0.5)$) -- (qc);
    
    \end{tikzpicture}
    \caption{A flower in a parity automaton. The flower loops are $q_c \xrightarrow{a} q_1 \xrightarrow{a} q_c$ for color 1, $q_c \xrightarrow{a} q_1 \xrightarrow{c} q_2 \xrightarrow{b} q_c$ for color 2, $q_c \xrightarrow{b} q_2 \xrightarrow{a} q_3 \xrightarrow{c} q_c$ for color 3, $q_c \xrightarrow{b} q_2  \xrightarrow{b} q_c$ for color 4, and $q_c \xrightarrow{c} q_3  \xrightarrow{c} q_c$ for color 5.}
    \label{fig:flowerVersionB}
\end{figure}
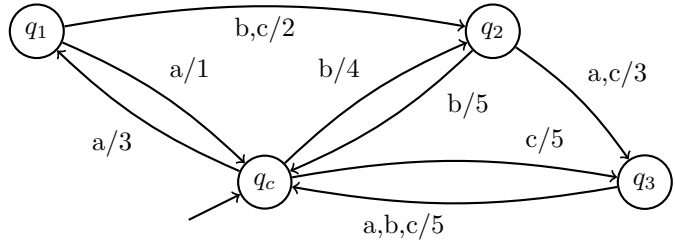

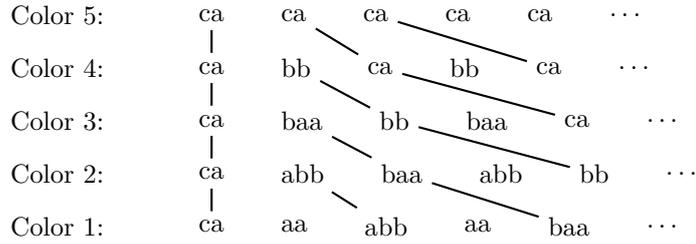
\begin{figure}
    \centering
    \begin{tikzpicture}
    \node[anchor=east] at (0,0) {Color 5:};    
    \node[anchor=east] at (0,-0.7) {Color 4:};
    \node[anchor=east] at (0,-1.4) {Color 3:};
    \node[anchor=east] at (0,-2.1) {Color 2:};
    \node[anchor=east] at (0,-2.8) {Color 1:};
    
    \node[anchor=west] (a1) at (1,0) {ca};
    \node[anchor=west] (a2) at ($(a1.east)+(0.5,0)$) {ca};
    \node[anchor=west] (a3) at ($(a2.east)+(0.5,0)$) {ca};
    \node[anchor=west] (a4) at ($(a3.east)+(0.5,0)$) {ca};
    \node[anchor=west] (a5) at ($(a4.east)+(0.5,0)$) {ca};
    \node[anchor=west] (a6) at ($(a5.east)+(0.5,0)$) {$\ldots$};
    
    \node[anchor=west] (b1) at (1,-0.7) {ca};
    \node[anchor=west] (b2) at ($(b1.east)+(0.5,0)$) {bb};
    \node[anchor=west] (b3) at ($(b2.east)+(0.5,0)$) {ca};
    \node[anchor=west] (b4) at ($(b3.east)+(0.5,0)$) {bb};
    \node[anchor=west] (b5) at ($(b4.east)+(0.5,0)$) {ca};
    \node[anchor=west] (b6) at ($(b5.east)+(0.5,0)$) {$\ldots$};
    
    \node[anchor=west] (c1) at (1,-1.4) {ca};
    \node[anchor=west] (c2) at ($(c1.east)+(0.5,0)$) {baa};
    \node[anchor=west] (c3) at ($(c2.east)+(0.5,0)$) {bb};
    \node[anchor=west] (c4) at ($(c3.east)+(0.5,0)$) {baa};
    \node[anchor=west] (c5) at ($(c4.east)+(0.5,0)$) {ca};
    \node[anchor=west] (c6) at ($(c5.east)+(0.5,0)$) {$\ldots$};
    
    \node[anchor=west] (d1) at (1,-2.1) {ca};
    \node[anchor=west] (d2) at ($(d1.east)+(0.5,0)$) {abb};
    \node[anchor=west] (d3) at ($(d2.east)+(0.5,0)$) {baa};
    \node[anchor=west] (d4) at ($(d3.east)+(0.5,0)$) {abb};
    \node[anchor=west] (d5) at ($(d4.east)+(0.5,0)$) {bb};
    \node[anchor=west] (d6) at ($(d5.east)+(0.5,0)$) {$\ldots$};
    
    \node[anchor=west] (e1) at (1,-2.8) {ca};
    \node[anchor=west] (e2) at ($(e1.east)+(0.5,0)$) {aa};
    \node[anchor=west] (e3) at ($(e2.east)+(0.5,0)$) {abb};
    \node[anchor=west] (e4) at ($(e3.east)+(0.5,0)$) {aa};
    \node[anchor=west] (e5) at ($(e4.east)+(0.5,0)$) {baa};
    \node[anchor=west] (e6) at ($(e5.east)+(0.5,0)$) {$\ldots$};
    
    \draw[thick] (a1) -- (b1);
    \draw[thick] (a2) -- (b3);
    \draw[thick] (a3) -- (b5);
    
    \draw[thick] (b1) -- (c1);
    \draw[thick] (b2) -- (c3);
    \draw[thick] (b3) -- (c5);

    \draw[thick] (c1) -- (d1);
    \draw[thick] (c2) -- (d3);
    \draw[thick] (c3) -- (d5);
    
    \draw[thick] (d1) -- (e1);
    \draw[thick] (d2) -- (e3);
    \draw[thick] (d3) -- (e5);

    \end{tikzpicture}
    \caption{{An example hierarchy of words for the parity automaton from Figure~\ref{fig:flowerVersionB}, used in an example in Section~\ref{sec:TowardsCanonicalRepresnetaion}. }
    {The words from the lower colors are obtained from those of higher colors by inserting language invariant words, which are flower loops in the given example.
    The lines show where the flower loops from the words with higher color are in the words with lower color.}}
    \label{fig:wordsInTheHierachy}
\end{figure}

\subsection{The case of the most significant color}

Let us now distill the definition of acceptance with the most significant color to the language case. {In the resulting definition,} the restriction to a single SCC will also be lifted.

The hierarchy of words from Figure~\ref{fig:wordsInTheHierachy} refers to the states of a given parity automaton: Finite words can only be inserted at places, in which the added finite words loop from the state in which the run of the automaton is at the insertion place, back to the same state.
This idea can be lifted by replacing the notion of state by \emph{suffix language invariance}. 
We say that inserting a finite word $u \in \Sigma^*$ at position $i \in \NN$ in a word $w = w_0 w_1 \ldots \in \Sigma^\omega$ is a suffix language invariant injection into $w$ for a language $L \subseteq \Sigma^\omega$ if 
$\mathcal{L}^\mathit{suffix}(L,w \ldots w_i) = \mathcal{L}^\mathit{suffix}(L,w \ldots w_i u)$.

In this definition, the suffix languages take the role of the flower center states, and they have the nice property that they are independent of an automaton representation of the language. 

Injecting any finite number of loops from a state back to itself to a run of a parity automaton does not change the color with which the respective word is accepted. 
It makes sense to expect for the definition of the natural color of a word that similarly, any finite number of suffix language invariant injections should not change the color of a word.

With this in mind, we can try to liberate the definition regarding which words should be recognized with the most significant color from any reference to a particular automaton:
they are those words for which an infinite number of suffix language invariant word injections does not change whether or not a word is accepted.
It is, however, necessary to carefully define where exactly in a word these injections can be made.

To see this, consider the case of the language `there are infinitely often two $a$s in a row'.
It can be recognized by a deterministic \emph{Büchi} automaton, i.e., a deterministic parity automaton with colors $0$ and $1$.
Since $0$ is the most significant color, all words in the language need to have this natural color.
This language has only a single suffix language, namely itself.
If we would require for a word to be of natural color $0$ that \emph{all} infinite sequences of suffix language invariant word injections result in an accepted word, then \emph{no} word would be accepted with this natural color: This is because this would include injecting a $b$ as every second character in a word.
Consequently, not all injection sequences need to be tolerated for a word to be of natural color $0$.
But it also does not suffice if only \emph{some} injection is tolerated:
In an extreme case, that includes injecting the empty word everywhere, which never affects acceptance.

A compromise between these two extremes is to use different quantifiers for the points of injection and the words being injected.
We declare those words to have a lowest natural color, for which 
\emph{there exists} an infinite sequence of points, at which suffix invariant words can be injected, such that, \emph{for all} insertions of sequences of suffix invariant words at these points, the resulting word is accepted if, and only if, the original word was.
While this solves the problem from the short example in the paragraph above, it is not trivially clear whether or not other problems remain when using this definition. 
The correctness of the construction from the next section, however, shows that this is precisely the definition we need.

\subsection{Generalizing to all colors}
To generalize the idea of the natural color of a word from the most significant color to the general case, we can follow an inductive argument and---in a sense---peel the language off, layer by layer.
We look at the colors $c \in \NN$ in ascending order and define, for each color $c+1$, which words are natural for this color, under the assumption that we have already defined this for colors up to $c$.
To do so, we can marry an inductive definition of what constitutes the color of a word in a parity automaton with the automaton-agnostic definition of the natural color of a word for the most significant color from above. %

We start with revisiting an inductive definition of the color of a word in a parity automaton.
We can characterize the words accepted with color $0$ to be those along whose runs transitions with color $0$ are taken infinitely often.
For colors $c>0$, we can define that a word is accepted with color $c$ if transitions with color $c$ are taken infinitely often for its run, \emph{and} the word is not accepted with a color smaller than $c$.
The nice property of this rather indirect definition is that it only refers to colors already defined and a single additional color. 

In an orthogonal composition of this idea for an inductive definition with the central idea from the previous subsection, we can allocate the color $c$ as the natural color of a word $w$ if there exists an infinite set of indices such that, for every sequence of suffix invariant strings inserted at these indices, we have that the resulting word $\tilde w$
\begin{itemize}
    \item \emph{either} has a natural color smaller than $c$,
    \item \emph{or} it does not, and is in the language if, and only if, $w$ is in the language.
\end{itemize}
Thus, we only require that inserting the words makes no difference regarding acceptance where the resulting words are not of a smaller natural color.

This definition has the nice property that, by induction, for every color, the natural colors of words are uniquely defined purely by the language of the word, without reference to an automaton representation.
The concrete definition {given in the next section}, however, makes the words for colors $c+i$ (for $i \geq 0$) contained in the words for a color $c$, so we define the natural color to be the \emph{minimal} color to which this definition is applicable. %

While it is, again, not trivially clear that this idea captures all necessary aspects of the natural color of a word, our results in the following section show that this is the case.

\section{The Natural Color of a Word}
\label{sec:wordColor}

{In this section, we define our notion of the natural color of a word (w.r.t.~a given language) based on the observations from the previous section. Based on this definition, we show how a sequence of co-Büchi automata can be obtained from the parity automaton that, after their minimization, is a canonical representation of the language. This construction has two steps, namely first \emph{streamlining} the determinstic parity automaton, and then extracting the co-Büchi automata from it.
Using this sequence of automata as representations for the natural colors of words (w.r.t.~the represented language), we finally show that the natural colors of words can %
be read off from the streamlined deterministic parity automaton directly, which %
shows that the natural colors of words can be read off from \emph{any} deterministic parity automaton after preprocessing it.
}

We start with defining the colors in which a word is `at home' {for a given language}. {The natural color of a word is then the minimal color at which a word is at home.}
First, we repeat some concepts from the previous section to make Definition~\ref{def:beingAtHome} below self-contained.
We say that some finite word $u \in \Sigma^*$ is \emph{suffix language invariant} for a
language $L$ after a finite word $w \in \Sigma^*$ if we have $\mathcal{L}^\mathit{suffix}(L,w) = \mathcal{L}^\mathit{suffix}(L,w u)$.

Let a word $w = w_0 w_1 \ldots \in \Sigma^\omega$  over an alphabet $\Sigma$ be given. We say that some word $w' \in \Sigma^\omega$ is the result of a suffix language invariant injection of a sequence of words $u_0,u_1,\ldots$ at positions $J= \{i_0,i_1,\ldots\}$ with $i_0<i_1<\ldots$ in $w$ if and only if $w' = w_0 w_1 \ldots w_{i_0} \, u_0 \,  w_{i_0+1} \ldots w_{i_1} \, u_1 \, w_{i_1+1} \ldots w_{i_2} \, u_2 \ldots$ and for each $j \in \NN$, we have that $u_j$ is suffix language invariant for $w_0 \ldots w_{i_j}$.

\begin{definition}
\label{def:beingAtHome}
For every even/odd $i$, we say that a %
word
$w = w_0 w_1 \ldots \in \Sigma^\omega$ is \emph{at home in color $i \in \mathbb{N}$} for some language $L$ if there exists an infinite subset $J \subseteq \mathbb{N}$
such that, for every possible sequence of finite words $u_0, u_1, \ldots$, if 
a word $w'$ is the result of a suffix language invariant injection of $u_0, u_1, \ldots$ at positions $J$, then
we have that
\begin{itemize}
\item $w'$ is already at home in a color strictly smaller than $i$, or
\item both $w$ and $w'$ are in $L$ and $i$ is even, or both $w$ and $w'$ are not in $L$ and $i$ is odd.
\end{itemize}
\end{definition}

Note that the first case in the preceding definition cannot apply for color $i=0$ (as there is no smaller color), so only the second case is of relevance for $i=0$.

We call the minimal natural number that a word $w$ is at home in the \emph{natural color of $w$} (for a given language $L$).

As a small remark, 
this definition could also be given in a variant where the lowest color a word can be at home in is $1$, which swaps the order in which we check for `$i$ is even' and `$i$ is odd'.
This affects the number of different natural colors that the words can have (for a given language) by at most $1$.
All results given henceforth also work with odd and even colors swapped. %

The color of a run of a DPA is not necessarily the natural color of the word defined above (for the language of the DPA).
We will, however, show how a deterministic parity automaton can be used to define a family of good-for-games co-Büchi automata that can determine the natural color of a word.

\subsection{Streamlining DPAs}

{The first concept to use is the concept of a  \emph{structured parity automaton}:} 
We call a DPA \emph{structured} if
(1) all states of $\mathcal A$ are reachable and
(2) if two states $q$ and $q'$ are equivalent, then they are in the same maximal SCC.

Turning a given DPA into an equivalent structured DPA is cheap.%
\footnote{The procedure can for instance be found in \cite{DBLP:conf/fsttcs/Schewe10}; while it is only described for B\"uchi and co-B\"uchi automata there, it can also be applied to parity automata, as done in \cite{DBLP:conf/csl/LodingT20}.} {First, non-reachable states are removed. Then,  an arbitrary minimal preorder that preserves reachability among the maximal SCCs of the DPA is defined.}
Two states are equivalent according to this preorder if, and only if, they are reachable from each other (i.e., they are in the same maximal SCC).
Apart from this, the preorder follows the reachability relation in that, if a state $q$ is reachable from a state $q'$, then $q \geq q'$.

A transition to a state $q$ with a language equivalent state $q'$ s.t.\ $q' > q$ is then re-directed to some language equivalent state $q''$ that is maximal according to this preorder.

As the position in this preorder can only grow along every run, there are only finitely many re-directions taken on every run.
The language of the automaton is not changed by this operation. 

{Finally, the states that become unreachable by rerouting the transitions are removed again to make the resulting automaton structured.} 

\begin{definition}
\label{def:streamlinedDPA}
Let a structured DPA $\mathcal A = (Q,\Sigma,\delta,q_0)$ be given.
We define its \emph{streamlined} version to be the outcome of the following \emph{streamlining process}, in which the structure of $\mathcal A$ is not changed, but a new color is assigned to each transition.

We first produce a copy of the structure of $\mathcal A$, creating a fresh \emph{coloring graph} $\mathcal G= (Q,\Sigma,\delta)$ (which is called `coloring graph' because it is used for determining the new colors in $\mathcal A$; $\mathcal G$ itself does not have colors).
We will then successively remove states and transitions from $\mathcal G$ (not from $\mathcal A$), while assigning the transitions {new colors} in $\mathcal A$.

Starting with $i=0$, we do the following until $\mathcal G$ is empty.
\begin{enumerate}
    \item We first partition $\mathcal G$ into maximal SCCs.
    \item We then identify all transient transitions in $\mathcal G$.
    We change their colors in $\mathcal A$ to $i$,
    and then remove the transient states and transitions from $\mathcal G$.
    \item 
    {We check if in any SCC of $\mathcal{G}$, the least color of the transitions (in $\mathcal{A}$) between states in the SCC has the same parity as $i$.
    \begin{itemize}
    \item 
    If such transitions are found, we first change their colors in $\mathcal A$ to $i$, remove them from $\mathcal G$, and then return to (1) without incrementing $i$.
    \item If there are no such transitions, we increment $i$ and go back to (1).
    \end{itemize}}
\end{enumerate}
These steps are repeated until $\mathcal G$ is empty.
\end{definition}
{The purpose of the construction is to iteratively \emph{lower} the colors of transitions in $\mathcal{A}$ towards the most significant one whenever that is possible without changing the language of the automaton. The \emph{structure} of the automaton is not altered.
The algorithm identifies transitions that can be recolored to a lower color $i$ because they can only be taken finitely often along runs that do not have a dominating color lower than or equal to $i$ anyway.
Also, the algorithm identifies transitions that can be safely recolored to $i$ because, while changing the color to $i$ changes the dominating color of some runs, it never changes the parity of the dominating color.}

The streamlining construction above is a variant of an algorithm by Carton and Maceiras \cite{DBLP:journals/ita/CartonM99} to relabel the colors of states in a parity automaton with state-based acceptance. The construction has been adapted to the case of transition-based acceptance and gives rise to a more concise correctness argument stated next.
{It creates an automaton, where all language equivalent states are in the same SCC.}

\begin{lemma}
The streamlining process terminates and does not change the language of $\mathcal A$. 
\end{lemma}

\begin{proof}
We first observe that, by a simple inductive argument, before $i$ is incremented, all transitions with color $\leq i$ have been removed from $\mathcal G$.

As induction basis, for $i=0$, every transition with color $0$ that remains after step (2) is the minimal color in {their} maximal SCC, and therefore removed in step (3).

For the induction step, after incrementing the counter to $i+1$, all transitions with color $\leq i$ have been removed by the induction hypothesis, such that all remaining transitions with color $i+1$ must either be transient (and removed in step (2)), or of minimal color in a maximal SCC (and are then removed in step 3).

{The algorithm always terminates because the color of a transition is only ever changed once (as the transition is removed from $\mathcal{G}$ whenever their color in $\mathcal{A}$ is changed){. When no color is (re)assigned in $\mathcal A$ and removed from $\mathcal G$ in an iteration}, $i$ is increased. {Finally, o}nce $i$ exceeds the number of colors of $\mathcal{A}$, the graph $\mathcal{G}$ {must be} empty, leading to termination}.

{The induction argument above also} establishes that the color of each transition can only be \emph{reduced} by this construction, but to no color lower than any (new) color of the edges that have previously been removed from $\mathcal G$.

This observation is the basis for establishing language equivalence.
We establish this language equivalence step-wise, considering only {the effect that} the changes of colors initiated by step (2) or step (3) in one iteration have {on the dominating color of \emph{some} (arbitrary) run of $\mathcal{A}$}.

For colors changed in step (2) of the algorithm, we observe that, if the respective transition $t$ occurs infinitely often, some other transition that has previously been removed from $\mathcal{G}$ must occur infinitely often, too.
The color of any of these removed transitions is $\leq i$, such that changing the color of $t$ to $i$ does not change the dominating color of the run.
If $t$ {however occurs only finitely often along the run in $\mathcal{A}$}, it cannot change the dominating color of the run either.

For a transition $t$ whose color is changed in step (3), we {distinguish three cases.}
First, if $t$ occurs only finitely often {along the run under %
{consideration}}, the color change does not influence the dominating color of the run.
Second, if {the run eventually remains} in the same maximal SCC (in $\mathcal{G}$) as $t$ was and $t$ occurs infinitely often, then the previous color of $t$ was the dominating color of the run, because $t$'s color was minimal among the colors of the SCC.
But then the new dominating color is $i$, which has the same parity as (and is no greater than) the previous color of $t$.
Finally, {if the run does not eventually get stuck in the maximal SCC of $t$ in $\mathcal{G}$}, but $t$ occurs infinitely often, then there are infinitely many transitions passed that have been re-colored before $t$, and that therefore have a color $\leq i$, such that the re-coloring of $t$ does not change the dominating color of the run.
\end{proof}

\subsection{From Streamlined DPAs to Color-Recognising GCAs}

We will not relate the colors of the runs from streamlined automata directly to the natural color of a word, but use them to define good-for-games co-Büchi automata {(\emph{GCAs})} that do so.
These GCAs are quite easy to obtain from $\mathcal A$.

\begin{definition}
\label{def:individualAutomata}
Let $\mathcal A=(Q,\Sigma,\delta,q_0)$
be a streamlined automaton and $i$ be a color that occurs in $\mathcal A$, then $\mathcal A_i = (Q,\Sigma,\delta_i,q_0)$ is the automaton such that, for all $(q,x,q',c) \in \delta$,
\begin{itemize}
\item $(q,x,q',2)\in \delta_i$ if $c \geq i$,
\item $(q,x,q',1)\in \delta_i$ if $c < i$, and
\item $(q,x,q'',1)\in \delta_i$
{for all $q'' \in Q$ with $q' \sim_{\mathcal A}q''$ such that},
for all colors $c'$, $(q,x,q'',c') \notin \delta$.
\end{itemize}
\end{definition}

{The co-Büchi automata are defined such that, for all colors $i$, $\mathcal{A}_i$ accepts those words, for which the run in $\mathcal{A}$ has a dominating color of at most $i$ (using transitions of the first two types in the list above only). However, $\mathcal{A}_i$ accepts some additional words: The transitions added by the third item in the list above allow a run to ``jump'' to any state that is language-equivalent in $\mathcal{A}$ (if the transition is not already part of $\mathcal{A}_i$ by the first two items). In accepting runs, this can only happen finitely often, though.

We will show in the remainder of this subsection that, for all $i$, $\mathcal{A}_i$ is a good-for-games co-Büchi automaton that accepts exactly the words that are at home in color $i$ (w.r.t.~the language of $\mathcal{A}$).}

The first observation {that we will use to achieve this goal} is that the languages of these co-Büchi automata are obviously shrinking with growing index, simply because the transitions are the same, but some of the accepting transitions become rejecting transitions (i.e., their color changes from $2$ to $1$).

\begin{observation}
\label{obs:inclusion}
For $i \leq j$, $\mathcal L(\mathcal A_i) \supseteq \mathcal L(\mathcal A_j)$ holds. Also, $\mathcal{A}_0$ accepts all words as it only has accepting transitions (since no color is smaller than $0$) while including outgoing transitions for each state/character pair (as $\mathcal{A}$ is deterministic).
\end{observation}

\begin{theorem}\label{theo:AiGFG}
For all $i$, $\mathcal A_i$ is good-for-games.
\end{theorem}

The proof is a pretty standard proof for co-Büchi automata: it merely says `follow the run that has longest been through accepting transitions'.

\begin{proof}
If the automaton $\mathcal A_i$ has no accepting run for a word $w = w_0 w_1 w_2 \ldots$, there is nothing to show.

We now assume that $w$ has an accepting run $\pi = q_0 q_1 q_2 \ldots$, where $j$ is the first position in $\pi$ such that, for all $k \geq j$, $(q_k,w_k,q_{k+1},2)$ are accepting transitions. %

We argue that $\mathcal A_i$ can accept $w$ with the strategy to
\begin{enumerate}
    \item follow accepting transitions where possible; note that there is at most one outgoing accepting transition for every state / input letter pair, so this selection is deterministic, and
    \item if no such transition is available when reading $w_k$, move to a state $q'_{k+1}$ such that there exists a run prefix $\pi' = q_0 q_1' q_2' \ldots q_{k+1}'$ for $w_0 \ldots w_k$ with some $l < k+1$ such that the transitions taken from $q'_l$ onwards are all accepting. In particular, state $q'_{k+1}$ is chosen for some \emph{lowest possible} value of $l$ among such run prefixes (the way to choose ex aequo does not matter). 
\end{enumerate}

{
To see why this strategy yields an accepting run, consider a total order over all possible finite run prefixes. The prefixes are ordered by their size (starting from the smallest one), but otherwise the total order is arbitrary. 

{Whenever the strategy can continue along an accepting transition, it does so.
When it has to take a rejecting transition after having read a finite prefix $w'$ of $w$, it chooses a smallest run prefix $\rho'$, such that $w'$ has a unique finite run $\rho'\rho''$, where $\rho''$ contains only accepting transitions.
$\rho'\rho''$ ends in some state $q$, and our strategy is to move to this state $q$.}

The prefix $q_0 \ldots q_{{j}}$ %
is somewhere in this order, {say at position $p$. Now, if there are at least $p$ rejecting transitions taken when following the strategy, then $q_0 \ldots q_{{j}}$ will eventually be tried. From this point onward}, no rejecting transitions are taken anymore in the run for $w$. If fewer than $p$ rejecting transitions are taken when following the strategy, the resulting run is accepting as well.
}

Thus, we always have a strategy that only relies on the past, and $\mathcal A_i$ is good-for-games.
\end{proof}

\begin{theorem}
$\mathcal A_i$ accepts a word $w$ if, and only if, its natural color for the language $\mathcal L(\mathcal A)$ is at least $i$ {(i.e., $w$ is at home in color $i$)}.
\end{theorem}

\begin{proof}
\underline{Induction basis:}
For $i=0$, every word is accepted by $\mathcal A_0$.
\smallskip

\noindent%
\underline{Induction step:}
Let us assume that the property holds for all $i'< i$ for some $i>0$.

{For the induction step, we split the ``if and only if'' in the claim into its two directions.}
We first show that (\emph{substep 1}) a word $w = w_0w_1 \ldots$ with natural color {at most} $i-1$ is rejected by $\mathcal A_i$, and then argue that (\emph{substep 2}) a word rejected by $\mathcal A_i$ has a natural color of at most $i-1$. {Taking both directions together and considering the remaining words (those \emph{accepted} by $\mathcal{A}_i$ rather than those \emph{rejected} my $\mathcal{A}_i$), we obtain that a word is accepted by $\mathcal{A}_i$ if, and only if, its natural color is at least $i$, which is to be proven.} 

\underline{Substep 1:} {Here, we show that a word $w = w_0w_1 \ldots$ with natural color of at most $i-1$ is rejected by $\mathcal A_i$.
If the natural color of $w$ is strictly smaller than $i-1$, then this follows directly from the inductive hypothesis and Observation~\ref{obs:inclusion}. For the case of the natural color of $w$ being exactly $i-1$, we assume for contradiction that $w$ has a natural color of $i-1$ and $w$ is accepted by $\mathcal{A}_i$. Using the definitions for acceptance by a co-Büchi automaton and the natural color of a word, we have that}
\begin{itemize}
\item $\pi = q_0q_1q_2\ldots$ is an accepting run of $\mathcal A_i$ on $w$,
\item $p \in \NN$ is a position such that, for all $j\geq p$, $(q_j,w_j,q_{j+1},2)\in\delta_i$ is an accepting transition in $\mathcal A_i$, and
\item $J \subseteq \NN$ is an infinite index set such that injecting suffix language invariant words at the positions in $J$ always results in a word $w'$ that
\begin{description}
    \item[(c1)] has natural color {$< i-1$ \emph{or}}
    \item[(c2)] is accepted by $\mathcal A$ if, and only if, {$i-1$} is even,
\end{description}
where we assume 
w.l.o.g.\ $j \geq p$ for all $j \in J$.
\end{itemize}
The accepting run $\pi$ has, from position $p$ onward, only transitions that have color $\geq i$ {in $\mathcal A$.
Let $p' \geq p$ be the position from which these transitions are all in the same maximal accepting SCC $S$ in $\mathcal{A}_i$.
By the assumption that $\mathcal{A}$ is streamlined (which is a precondition to applying Definition~\ref{def:individualAutomata}),
the maximal accepting SCC (of $\mathcal{A}_i$) has a transition that has a corresponding transition of color $i$ in $\mathcal{A}$.
To see this, note that we can only have an accepting SCC in $\mathcal{A}_i$ if it is also an SCC in the graph $\mathcal{G}$ built by the streamlining construction when starting to consider color $i$ (as all transitions from and to states not in $\mathcal{G}$ {at that point of the construction have been assigned colors strictly smaller than $i$} in the streamlined automaton).
But then, either the minimal transition color in the SCC has the same parity as $i$, and then it is lowered to $i$ in the streamlining construction, or it does not. In the latter case, the streamlining construction lowers the color of such a minimal transition color in the SCC to $i-1$ by the third step of the construction before actually considering color $i$, contradicting the assumption that the SCC has only accepting transitions in $\mathcal{A}_i$ (as transitions with color smaller than $i$ are not accepting in $\mathcal{A}_i$ by Def.~\ref{def:individualAutomata}).

Since we now know that the minimal color in the maximal accepting SCC in $\mathcal{A}_i$ is $i$} {in the streamlined automaton $\mathcal A$,}
we can insert into $w$, in every position in $j \in J$, a suffix language invariant string, whose partial run is a cycle in both $\mathcal{A}_i$ and $\mathcal{A}${, in the latter case with minimal color $i$}.
Therefore, the resulting word {$w'$} is accepted by $\mathcal A$ if, and only if, $i$ is even.
{As $i$ and $i-1$ have a different parity, condition (c2) cannot hold for $w'$.

However, condition (c1) also cannot hold: The accepting run of $\mathcal A_i$ on $w'$ is also an accepting run of $\mathcal A_{i-1}$, so its natural color is at least $i-1$ by our inductive hypothesis.
Taking the falsification of both (c1) and (c2) together, we obtain that $w$ cannot have a natural color of at most $i-1$.} \hfill 
 (contradiction)
\smallskip

\underline{Substep 2:}

Finally, we show that a word $w = w_0w_1 \ldots$  that is rejected by $\mathcal A_i$ has natural color of at most $i-1$. 
{If the word is rejected by $\mathcal{A}_{i-1}$, then the natural color is at most $i-2$ by our inductive hypothesis, and there is nothing more to be shown.
So we henceforth only need to consider the case that $w$ is accepted by $\mathcal{A}_{i-1}$ but not $\mathcal{A}_i$.
}
We define a suitable infinite set of indices $J\subseteq\NN$.

We first assume for contradiction that there is a position $p>0$ in the word such that there is no position $p'>p$ such that there is a run $\pi = q_0q_1 \ldots$ of $\mathcal A_i$ where $(q_j,w_j,q_{j+1},2)\in\delta_i$ is, for all $p\leq j <p'$, an accepting transition in $\mathcal A_i$.
If no such position exists, the finitely branching tree of runs of $\mathcal A_i$ on $w$ that is pruned at all non-accepting positions after level $p$ is infinite, and therefore has an infinite path.
\hfill \mbox{(contradiction to $w \notin \mathcal L(\mathcal A_i)$)}

Using this observation, we fix an infinite ascending chain $0 < p_0 < p_1 < p_2 \ldots$, such that, for all $j \geq 0$, no run has only accepting transitions in any segment $q_{p_j} q_{p_j+1}\ldots q_{p_{j+1}}$. 

We note that inserting suffix language invariant strings in positions of $J= \{p_j \mid j\in \NN\}$ does not change that these segments have this property; consequently, any word $w'$ that results from such insertions is still rejected by $\mathcal A_i$.
We fix such a word $w'$.

Let $c$ be the maximal color such that $w'= w_0'w_1'w_2'\ldots$ is in the language of $\mathcal A_c$. As $w' \notin \mathcal L(\mathcal A_i)$, $c<i$. If $c < i-1$, then its natural color is $c < i-1$ by induction hypothesis.
If it is $c=i-1$, then the natural color must be at least $i-1$ by induction hypothesis.
Moreover, $\mathcal A_{i-1}$ has an accepting run $\pi'=q_0q_1'q_2' \ldots$ on $w'$.
Let $p$ be a position in this run such that, for all $j>p$, the transitions $(q_j',w_j',q_{j+1}',2)\in\delta_{i-1}$ are accepting transitions of $\mathcal A_{i-1}$.
Noting that $\pi'$ is a rejecting run of $\mathcal A_i$, this entails that the lowest color that occurs infinitely often in the run $q_pq_{p+1}q_{p+2}\ldots$
of $\mathcal A_{q_p}$ on $w_pw_{p+1}w_{p+2}\ldots$ is $i-1$.
Thus $w_pw_{p+1}w_{p+2}\ldots$ is accepted by $\mathcal A_{q_p}$, and therefore $w'$ is accepted by $\mathcal A$ if, and only if, $i$ is odd.

This concludes the proof that the natural color of $w$ is $i-1$.
\end{proof}

{Taking the results above together, we have obtained a construction for {a language recognised by a given deterministic parity automaton that provides} a sequence of co-Büchi automata that encode which word is at home in which color:
We first compute a structured form of this deterministic parity automaton, then streamline it (Def.~\ref{def:streamlinedDPA}), and finally split the resulting parity automaton into the co-Büchi automata according to Def.~\ref{def:individualAutomata}. All three steps can be implemented to run in time polynomial in the size of the input automaton. Furthermore, since the state spaces of the co-Büchi automata are the same as the one in the parity automaton, they cannot be larger than the original parity automaton.

Since the split into co-Büchi automata is canonical, and the co-Büchi automata themselves can be made canonical (and minimal) using the existing polynomial-time construction from Abu Radi and Kupferman \cite{DBLP:conf/icalp/RadiK19}, we overall obtain a canonical representation of the language that the parity automaton we started with represents{. Moreover,} in can be computed in polynomial time.
}

\subsection{Reading off Natural Colors from Streamlined Automata}

{The construction so far has the property that it does not immediately provide a direct way of computing the natural color of a word (yet). Given a word, we can check which of the co-Büchi automata built according to Def.~\ref{def:individualAutomata} accepts the word to compute the natural color of a word (for the given language), but since they are good-for-games rather than deterministic, this is somewhat cumbersome. 

The results above however allow us to also define a more direct way to determining the natural color of a word (w.r.t~a given language), as we show below as a side-result. 
}

We define a \emph{co-run} of a deterministic automaton $\mathcal A = (Q,\Sigma,\delta,q_0)$ on a word
$w = w_0 w_1 w_2 \ldots$
with a run $\pi = q_0 q_1 q_2 \ldots$
as a sequence $\pi' = q_0 q_1 \ldots q_{p-1} q_p q_{p+1}' q_{p+2}' q_{p+3}' \ldots $ for some $p>0$, such that
$\pi'' = q_p' q_{p+1}' q_{p+2}' q_{p+2}' \ldots$ is the run of $\mathcal A_{q_p'}$ on the word $w' = w_p w_{p+1} w_{p+2} w_{p+2} \ldots$ for some state $q_p'$ that is language equivalent to $q_p$ ($q_p' \sim_{\mathcal A} q_p$).

The color of the set of co-runs for a word $w$ is defined to be the maximal color $c$ that occurs infinitely often on some co-run of $w$.

\begin{lemma}
Let $\mathcal A$ be a streamlined automaton. Then the color of the set of co-runs of a word $w$ is $c$ if $w$ is in the language of $\mathcal A_c$, but not in the language of $\mathcal A_{c+1}$.
\end{lemma}

\begin{proof}
A co-run of $\mathcal A$ on $w$ with dominating color $c$ is an accepting run of $\mathcal A_c$.

An accepting run $\pi' = q_0' q_1' \ldots q_{p-1}' q_p' q_{p+1}' q_{p+2}' q_{p+2}' \ldots$ of $\mathcal A_{c+1}$ has some position $p$ from which point onward only accepting transitions (which all have color $>c$ in $\mathcal A$) are taken.
$\mathcal A$ therefore has a co-run $\pi' = q_0 q_1 \ldots q_{p-1} q_p q_{p+1}' q_{p+2}' q_{p+2}' \ldots $, whose dominating color is $>c$.
\hfill (contradiction)
\end{proof}

By combining this lemma with the previous theorem, we get the following corollary.

\begin{corollary}
Let $\mathcal A$ be a streamlined automaton. Then the color of the set of co-runs of a word $w$ is its natural color for $\mathcal L(\mathcal A)$.
\end{corollary}

{Co-runs are closely related to the GCAs we have defined earlier.
The difference is that the ``new'' transitions to language equivalent states can be used only once along a run.
This allows for having a definition on the deterministic automaton (without falling back on good-for-games automata), and is therefore simpler.
It also binds the proofs together: the \emph{minimal prefix} from the proof of Theorem \ref{theo:AiGFG} corresponds to the shortest prefix, at which this single transition to a language equivalent state can be taken.
While this provides a more direct connection to the color, the restriction to taking these transitions at most once loses the good-for-games property: as a wrong decision cannot be corrected, access to the remainder of the run may be required.
This makes GCAs more attractive, as co-B\"uchi languages have canonical representatives.}

{Note that for an ultimately periodic word (i.e., a word of the form $w = uv^\omega$ for $u,v \in \Sigma^*$, the highest color among the co-runs can be computed in the time polynomial in the number of states, which allows reading off the natural color of a word from a (streamlined) DPA without building the canonical representation of the language of the DPA. }

\section{Related Work}

There already exists an indirect normal form of $\omega$-languages.
Every $\omega$-regular language can be represented as a deterministic finite automaton (DFA): This DFA accepts words of the form $u\$v$, for which the \emph{ultimately periodic word} $uv^\omega$ is in the $\omega$-language to be represented. Calbrix et al.~\cite{DBLP:conf/mfps/CalbrixNP93} showed how to compute such a DFA from a given nondeterministic Büchi automaton (to which a deterministic parity automaton is easy to translate). The minimized DFA for this \emph{lasso language} over finite words can then serve as a canonical representation of the $\omega$-language, as two automata representing $\omega$-regular languages encode the same language if, and only if, they accept the same ultimately periodic words.

{DFAs that capture ultimately periodic words have furthermore been refined to \emph{families} of DFAs
\cite{DBLP:conf/alt/AngluinF14} that can be more succinct and share the property to consist of multiple sub-automata with the sequences of co-Büchi automata that we define in this paper.

Such DFAs (or families of DFAs) however do not implement a core idea of automata,} namely to read a word letter-by-letter and to encode the relevant information about the letters of the word already read in a state. It is also unknown how the sizes of {such automata relates to the size of a minimal parity automaton representing the language. Finally, neither DFAs for lasso languages nor families of DFAs have a direct connection to the complexity of a language, as it is, for example, captured by the minimal number of colors used in deterministic parity automata. }

\section{Conclusion}
\label{sec:conclusion}

A classical question in the theory of automata is how to define \emph{canonical automata}.

In this paper, we have taken a step back, and looked at the question of how to define canonical automata for $\omega$-regular languages from a new angle.
For a language to be defined canonically by, say, a canonical deterministic parity automaton, a word first and foremost needs a natural color.
What should this color be?

Picking the flowers of Niwińkski and Walukiewicz \cite{DBLP:conf/stacs/NiwinskiW98} as a starting point, we have lifted the same principle to languages in a way that is oblivious to the automaton used.

For an $\omega$-regular language $L$, we look at a set of languages $L_0 \supset L_1 \supset L_2 \ldots \supset L_c$ s.t.
\begin{itemize}
    \item $L_0$ is the universal language;
    \item $L_1$ is the largest co-B\"uchi language contained in $\overline{L}\cap L_0$, closed under insertion;
    \item $L_2$ is the largest co-B\"uchi language contained in $L \cap L_1$, closed under insertion;
    \item $L_3$ is the largest co-B\"uchi language contained in $\overline{L}\cap L_2$, closed under insertion; etc.
\end{itemize}

The \emph{closure under insertion} refers to the existence of an infinite set of positions at which suffix language invariant words can be added such that \emph{if} the resulting word is in $L_i$ \emph{then} it is in $L$ if, and only if, $i$ is even.

This assigns a \emph{color} to each infinite \emph{word} $w$, both in the language and outside of it, purely defined by the maximal $i$ such that $w \in L_i$---we say that $w$ has a natural color of $i$ (for $L$).
As one would expect for a parity condition, $w \in L$ if, and only if, the natural color $i$ of $w$ is even.

The natural color of $w$ for $L$ is thus defined without reference to an automaton (or any other representation of the language).
Yet, $L$ is recognized by a parity automaton with maximal color $c$ if, and only if, $L_{c+1}$ is empty.
This sets the minimal maximal color in the automaton to the maximal $i$ such that $L_i$ is non-empty, which further connects this construction to deterministic automata.

We infer these languages by turning a single \emph{streamlined} deterministic parity automaton, which is cheap and easy to obtain from any deterministic parity automaton that recognizes $L$%
{: $L_i$ contains the set of words whose co-runs have  colors of at least $i$}.

{Beyond providing evidence \emph{that} a word is in the language, it also provides insight into \emph{why} it is part of this language by peeling off co-B\"uchi languages of accepted and rejected words layer by layer.}

Returning to our chain of languages, this answers the `why co-B\"uchi?' question that begs to be asked.
{Each $L_c$ is a co-B\"uchi language, which is an ideal basis for} a natural representation, because co-B\"uchi languages have recently obtained a canonical representation, albeit not for deterministic automata, but for \emph{good-for-games co-B\"uchi automata} with transition-based acceptance (GCAs).
We can use this to obtain a natural representation for the languages that allow us to identify the natural color of a word.

\bibliography{bib}

\end{document}